\documentclass[11pt]{article}

\usepackage[T1]{fontenc}
\usepackage{textcomp}
\usepackage{palatino}
\usepackage{mathpazo}
\usepackage{stmaryrd}

\usepackage{hyperref}
\hypersetup{pdfpagemode=UseNone}

\usepackage{amsfonts}
\usepackage{amssymb}
\usepackage{amsmath}
\usepackage{latexsym}
\usepackage{amsthm}
\usepackage{eepic}
\usepackage{sectsty}
\usepackage[usenames]{color}
\usepackage{gastex}

\newtheorem{theorem}{Theorem}
\newtheorem{lemma}[theorem]{Lemma}

\theoremstyle{definition}

\newcommand{\tinyspace}{\mspace{1mu}}
\newcommand{\microspace}{\mspace{0.5mu}}
\newcommand{\op}[1]{\operatorname{#1}}

\newcommand{\norm}[1]{\left\lVert\tinyspace#1\tinyspace\right\rVert}

\newcommand{\abs}[1]{\left\lvert\tinyspace #1 \tinyspace\right\rvert}

\newcommand{\tr}{\operatorname{Tr}}

\newcommand{\poly}{\operatorname{poly}}

\newcommand{\ip}[2]{\left\langle #1 , #2\right\rangle}

\def\({\left(}
\def\){\right)}
\def\I{\mathbb{1}}

\newcommand{\setft}[1]{\mathrm{#1}}
\newcommand{\lin}[1]{\setft{L}\left(#1\right)}
\newcommand{\density}[1]{\setft{D}\left(#1\right)}

\newcommand{\channel}[1]{\setft{C}\left(#1\right)}

\newcommand{\pos}[1]{\setft{Pos}\left(#1\right)}

\def\complex{\mathbb{C}}

\def\natural{\mathbb{N}}

\def \lket {\left|}
\def \rket {\right\rangle}
\def \lbra {\left\langle}
\def \rbra {\right|}
\newcommand{\ket}[1]{\lket\microspace #1 \microspace\rket}
\newcommand{\bra}[1]{\lbra\microspace #1 \microspace\rbra}

\newenvironment{mylist}[1]{\begin{list}{}{
	\setlength{\leftmargin}{#1}
	\setlength{\rightmargin}{0mm}
	\setlength{\labelsep}{2mm}
	\setlength{\labelwidth}{8mm}
	\setlength{\itemsep}{0mm}}}
	{\end{list}}

\newcommand{\class}[1]{\textup{#1}}

\newcommand{\reg}[1]{\mathsf{#1}}

\def\A{\mathcal{A}}

\def\E{\mathcal{E}}

\def\R{\mathcal{R}}

\def\Q{\mathcal{Q}}
\def\S{\mathcal{S}}

\def\yes{\text{yes}}
\def\no{\text{no}}

\begin{document}

\title{\bf Quantum interactive proofs with short messages}

\author{%
  Salman Beigi${}^\ast$
  \quad\quad
  Peter W.~Shor${}^\dagger$
  \quad\quad
  John Watrous${}^\ddagger$\\[4mm]
  ${}^\ast$%
  {\small\it Institute for Quantum Information}\\[-1mm]
  {\small\it California Institute of Technology}\\[-1mm]
  {\small\it and}\\[-1mm]
  {\small\it School of Mathematics}\\ [-1mm]
  {\small\it Institute for Research in Fundamental Sciences (IPM)}\\[2mm]
  ${}^\dagger$%
  {\small\it Department of Mathematics}\\[-1mm]
  {\small\it Massachusetts Institute of Technology}\\[2mm]
  ${}^\ddagger$%
  {\small\it Institute for Quantum Computing and School of Computer
  Science}\\[-1mm]
  {\small\it University of Waterloo}
}

\date{\today}

\maketitle

\begin{abstract}
  This paper considers three variants of quantum interactive proof
  systems in which short (meaning logarithmic-length) messages are
  exchanged   between the prover and verifier.
  The first variant is one in which the verifier sends a short message
  to the prover, and the prover responds with an ordinary, or
  polynomial-length, message;
  the second variant is one in which any number of messages can be
  exchanged, but where the combined length of all the messages is
  logarithmic;
  and the third variant is one in which the verifier sends
  polynomially many random bits to the prover, who responds with a
  short quantum message.
  We prove that in all of these cases the short messages can be
  eliminated without changing the power of the model, so the first
  variant has the expressive power of \class{QMA} and the second and
  third variants have the expressive power of \class{BQP}.
  These facts are proved through the use of quantum state tomography,
  along with the finite quantum de Finetti theorem for the first
  variant.
\end{abstract}

\section{Introduction} \label{sec:introduction}

The interactive proof system model extends the notion of efficient
proof verification to an interactive setting, where a computationally
unrestricted \emph{prover} tries to convince a computationally bounded
\emph{verifier} that an input string satisfies a particular fixed
property.
They have been studied extensively in computational complexity theory
since their introduction roughly 25 years ago
\cite{GoldwasserMR85,GoldwasserMR89,Babai85,BabaiM88}, and
as a result much is known about them.
(See \cite{AroraB09} and \cite{Goldreich08}, for instance, for further
discussions of classical interactive proof systems.)

Quantum interactive proof systems are a natural quantum computational
extension of the interactive proof system model, where the prover and
verifier can perform quantum computations and exchange quantum
information.
The expressive power of quantum interactive proofs is no different
from classical interactive proofs: it holds that 
$\class{QIP} = \class{PSPACE} = \class{IP}$, and therefore any problem
having a quantum interactive proof system also has a classical one
\cite{JainJUW09,LundFKN92,Shamir92}.
However, quantum interactive proof systems may be significantly more
efficient than classical interactive proofs in terms of the number of
messages they require, as every problem in $\class{PSPACE}$ has a
quantum interactive proof system requiring just three messages to be
exchanged between a prover and verifier \cite{KitaevW00}.
This is not possible classically unless 
$\class{AM} = \class{PSPACE}$, and this equality implies the collapse of the
polynomial-time hierarchy \cite{BabaiM88,GoldwasserS89}.

In this paper we consider quantum interactive proof systems in which
some of the messages are short, by which we mean that the messages
consist of a number of qubits that is logarithmic in the input
length.
Three particular variants of quantum interactive proofs with short
messages are considered.
The first variant is one in which the verifier sends a short message
to the prover, and the prover responds with an ordinary, or
polynomial-length, message.
We prove that this model has the expressive power of $\class{QMA}$.
The second variant is one in which any number of messages can be
exchanged between the prover and verifier, but where the combined
length of all the messages is logarithmic.
We prove that this model has the expressive power of $\class{BQP}$.
The third variant is one in which the verifier sends polynomially many
random bits to the prover, who responds with a short quantum message.
We prove that this model also has the expressive power of
$\class{BQP}$.
Thus, in each of these three cases, logarithmic-length messages are
effectively worthless and can be removed without changing the power of
the model.


One possible application of our work is to the design of new quantum
algorithms or \class{QMA} verification procedures.
Although we do not yet have interesting examples, we believe it is
possible that an intuition about quantum interactive proof systems
with short messages may lead to new problems being shown to be in
\class{BQP} or \class{QMA}, based on characterizations of the sort we
prove.

Observe that all of these three results are immediate in the classical case. For example, one can enumerate all 
logarithmic-length interactions between a verifier and prover in polynomial-time, so our second model, assuming that the verifier is classical, has the expressive power of $\class{P}$
(or $\class{BPP}$ in the presence of randomness). This argument, however, does not work in the quantum case. To explain the difference let us consider the following simplification of this model. Assume that instead of an arbitrary number of messages of logarithmic total length, there is only one logarithmic-size message allowed which is sent by the prover. This model is denoted by $\class{QMA}_{\log}$, and was known to be equal to $\class{BQP}$~\cite{MarriottW05}. Here we present another proof for this fact to illustrate the main ideas of the paper. Upon receiving a logarithmic-size message from the prover, the verifier applies a binary measurement $\{P_{\text{acc}}, P_{\text{rej}}\}$ to decide whether to accept or reject. Thus the acceptance probability is at most the maximum eigenvalue of $P_{\text{acc}}$. Although $P_{\text{acc}}$ acts on a logarithmic number of qubits, it is given by a polynomial-size circuit, so one cannot directly compute the matrix representation of $P_{\text{acc}}$ in polynomial-time. Nevertheless, using quantum process tomography we can perform the measurement $\{P_{\text{acc}}, P_{\text{rej}}\}$ on polynomially many known states, and somehow by taking the average of their outcomes compute an approximation of $P_{\text{acc}}$. Since the matrix representation of $P_{\text{acc}}$ has only polynomially many entries, this approximation can be arbitrarily tight. The next step is to simply find the maximum eigenvalue of this approximation. 

In this paper instead of applying quantum process tomography on a measurement, we perform quantum state tomography on the normalized Choi-Jamio{\l}kowski representation of the quantum channel corresponding to the measurement. These two approaches are equivalent, but the second one unifies the arguments in different sections. 

Besides quantum state tomography and Choi-Jamio{\l}kowski representation of quantum channels, finite quantum de Finetti theorem is another important tool in this work. Suppose that we are given some copies of a state and we want to verify that it is closed to some given state. Using quantum state tomography on these copies we can find an approximation of the unknown state and solve the problem. Assume now that we are not guaranteed that these copies are indeed copies of the same state; there can even be entanglement among different copies. To overcome these difficulties we use finite quantum de Finetti theorem to reduce the problem to the first case.

The remainder of this paper has the following organization.
Section~\ref{sec:preliminaries} discusses some of the background
information needed for the rest of the paper, including background on
the Choi-Jamio{\l}kowski representation of quantum channels, quantum
state tomography, finite quantum de Finetti theorem, and quantum interactive proof systems.
Sections~\ref{sec:first-message-short}, \ref{sec:all-messages-short},
and \ref{sec:second-message-short} then discuss the three variants of
quantum interactive proof systems with short messages described
above.

\section{Background} \label{sec:preliminaries}

We assume the reader is familiar with quantum information and
computation, including the basic quantum complexity classes
$\class{BQP}$ and $\class{QMA}$, simple properties of mixed states,
measurements, channels, and so on~\cite{KitaevSV02, NielsenC00}.
The purpose of the present section is to highlight background
knowledge on three topics, represented by the three subsections below,
that are particularly relevant to this paper.
These topics are: the Choi-Jamio{\l}kowski representation of quantum
channels, quantum state tomography, and quantum interactive proof
systems.

Before discussing these three topics, it is appropriate to mention a
few simple points of notation and terminology.
Throughout this paper we let $\Sigma = \{0,1\}$ denote the binary
alphabet, and for each $k\in\natural$ we write $\complex(\Sigma^k)$
to denote the finite-dimensional Hilbert space whose standard basis vectors are indexed by
$\Sigma^k$ (i.e., the Hilbert space associated with a $k$-qubit
quantum register).
The Dirac notation is used to describe vectors in spaces of this sort.

For a given space $\Q = \complex(\Sigma^k)$, we write $\lin{\Q}$
to denote the space of all linear mappings from $\Q$ to itself, which
is associated with the space of all complex matrices with rows and
columns indexed by $\Sigma^k$ in the usual way.
The subsets of this space representing the positive semidefinite
operators and density operators on $\Q$ are denoted $\pos{\Q}$ and
$\density{\Q}$, respectively.
A standard inner product on $\lin{\Q}$ is defined as
$\ip{X}{Y} = \tr(X^{\ast}Y)$ for all $X,Y\in\lin{\Q}$ (where
$X^{\ast}$ denotes the adjoint, or conjugate-transpose, of $X$).
The trace norm of an operator $X\in\lin{\Q}$ is defined as
\[
\norm{X}_1 = \tr\sqrt{X^{\ast} X},
\]
and the spectral (or operator) norm of $X$ is denoted $\norm{X}$.

\subsection{Quantum channels and the Choi-Jamio{\l}kowski
  representation} \label{sec:Choi}

A \emph{quantum channel} from a $k$-qubit space 
$\Q = \complex(\Sigma^k)$ to an $l$-qubit space 
$\R = \complex(\Sigma^l)$ is a completely positive and
trace-preserving linear mapping of the form
$\Phi:\lin{\Q} \rightarrow\lin{\R}$. ($\Phi$ is completely positive if $\Phi\otimes I_{\lin{\S}}$, for every Hilbert space $\S$, is positive, meaning that it sends 
positive semidefinite operators to positive semidefinite ones. Trance-preserving means that $\tr(\Phi(\rho))= \tr(\rho)$.)
We will write $\channel{\Q,\R}$ to denote the set of all such quantum
channels.
For any quantum channel $\Phi\in\channel{\Q,\R}$ one defines the
(normalized) Choi-Jamio{\l}kowski representation~\cite{Jamiolkowski72, Choi75} of $\Phi$ as
\begin{equation} \label{eq:normalized-Choi}
  \rho = \frac{1}{2^k} \sum_{y,z\in\Sigma^k}
  \Phi(\ket{y}\!\bra{z}) \otimes\ket{y}\!\bra{z}.
\end{equation}
In other words, this is the $l+k$ qubit state that results from
applying $\Phi$ to one-half of $k$ pairs of qubits in the
$\ket{\phi^+} = (\ket{00} + \ket{11})/\sqrt{2}$ state.

The action of the mapping $\Phi$ can be recovered from its normalized
Choi-Jamio{\l}kowski representation in the following way that makes
use of \emph{post-selection}.
Suppose that $\reg{Q}$ and $\reg{Q}_0$ are $k$-qubit registers and
$\reg{R}$ is an $l$-qubit register, that the pair $(\reg{R},\reg{Q}_0)$
is initialized to the state $\rho$ as defined by $\Phi$ in
\eqref{eq:normalized-Choi}, and that $\reg{Q}$ is in an arbitrary
quantum state (and is possibly entangled with additional registers not
including $\reg{Q}_0$ and $\reg{R}$).
Consider the following procedure:
\begin{mylist}{\parindent}
\item[1.]
  Measure each qubit of $\reg{Q}$ together with its corresponding
  qubit in $\reg{Q}_0$ with respect to the Bell basis.
\item[2.]
  If every one of these $k$ measurements results in an outcome
  corresponding to the Bell state $\ket{\phi^+}$, then output
  ``success,'' else output ``failure.''
\end{mylist}

This procedure gives the outcome ``success'' with probability
$4^{-k}$, and conditioned on success the register $\reg{R}$ is
precisely as it would be had it resulted from the channel $\Phi$ being
applied to $\reg{Q}$. 
(The registers $\reg{Q}$ and $\reg{Q}_0$ can safely be discarded if
the procedure succeeds.)
To see this, assume first that the joint state of
$(\reg{R},\reg{Q}_0,\reg{Q})$ is $\rho\otimes\xi$ before the
measurement takes place.
Then the (unnormalized) state of $\reg{R}$ after the measurements are
performed, assuming the end result is ``success,'' is
\[
\frac{1}{2^{2k}}\sum_{y,y',z,z'\in\Sigma^k}
\Phi(\ket{y}\!\bra{z}) \langle y' | y\rangle \langle z|z'\rangle 
\langle y' | \xi | z'\rangle
= \frac{1}{4^k}\sum_{y,z\in\Sigma^k}
\Phi\(\ket{y}\!\bra{y}\xi\ket{z}\!\bra{z}\)
= \frac{1}{4^k}\Phi(\xi).
\]
The probability of success is therefore $4^{-k}$, and conditioned on
this outcome the process implements the channel $\Phi$. In our applications  $k$ is logarithmic 
in the size of the problem, so $\Phi$ is implemented with an inverse polynomial probability which is enough for us. 
The fact that this process implements the channel $\Phi$ exactly
for all density operators $\xi$ implies that it also operates correctly
in the case that $\reg{Q}$ is entangled with additional registers.

\subsection{Quantum state tomography} \label{sec:tomography}

Quantum state tomography is the process by which an approximate
description of an unknown quantum state is obtained by measurements on
many independent copies of the unknown state.
To be more precise, let $\Q = \complex(\Sigma^k)$ denote the space
corresponding to a $k$-qubit register, and suppose that
$\reg{X}_1,\ldots,\reg{X}_N$ are $k$-qubit quantum registers
independently prepared in an unknown $k$-qubit state
$\rho\in\density{\Q}$.
The purpose of quantum state tomography is to obtain an explicit
description of a $k$-qubit state that closely approximates $\rho$.

One way to perform quantum state tomography is through the use of an
\emph{information-complete measurement}.
A measurement $\{P_a\,:\,a\in\Gamma\}$ on $k$-qubit registers is
information-complete if and only if the set
$\{P_a\,:\,a\in\Gamma\}$ spans the entire $4^k$-dimensional space
$\lin{\Q}$.
When such a measurement is performed on a $k$-qubit state $\rho$, each
measurement outcome is obtained with probability
\[
p(a) = \ip{P_a}{\rho}.
\]
Based on the assumption that $\{P_a\,:\,a\in\Gamma\}$ is
information-complete, this vector $p$ of probabilities uniquely
determines the state $\rho$.
A close approximation of $p$, which may be obtained by sufficiently
many independent measurements, leads to an approximate description of
$\rho$.

The accuracy of an approximation based on the process just described
naturally depends on the choice of an information-complete measurement
as well as the specific notion of approximation that is considered.
Our interest will be on the trace distance $\norm{\rho - \sigma}_1$
between the approximation $\sigma$ and the true state $\rho$.
To describe the ``quality'' of an information-complete measurement, it
is appropriate to describe the specific process that is used to
reconstruct $\rho$ from the vector of probabilities $p$.

For any spanning set $\{P_a\,:\,a\in\Gamma\}$ of $\lin{\Q}$, there
exists a set $\{M_a\,:\,a\in\Gamma\}\subseteq\lin{\Q}$ that satisfies
\[
\sum_{a\in\Gamma}M_a \ip{P_a}{X} = X
\]
for every $X\in\lin{\Q}$.
(One may find such a set $\{M_a\,:\,a\in\Gamma\}$ by solving a
system of linear equations.)
The set $\{M_a\,:\,a\in\Gamma\}$ is uniquely determined when
$\{P_a\,:\,a\in\Gamma\}$ has exactly $4^k$ elements (i.e., is a
basis), and hereafter we will restrict our attention to this case. Notice that if $\rho$
is a density matrix, the coefficients $p(a)=\ip{P_a}{\rho}$ form a probability distribution.
If $q$ is a probability vector that represents an approximation to $p$, it holds that
\[
\norm{
\sum_{a\in\Gamma} p(a) M_a - 
\sum_{a\in\Gamma} q(a) M_a}_1
\leq \sum_{a\in\Gamma}\abs{p(a) - q(a)} \norm{M_a}_1
\leq \norm{p-q}_1\,\max_{a\in\Gamma}\norm{M_a}_1.
\]
It is therefore desirable that the maximum trace norm over the set
$\{M_a\,:\,a\in\Gamma\}$ determined by the measurement
$\{P_a\,:\,a\in\Gamma\}$ is as small as possible.

There is one additional consideration that is sometimes relevant,
which is that the approximation
\[
\sum_{a\in\Gamma} q(a) M_a
\]
may fail to be positive semidefinite, and therefore fail to represent
a valid quantum state.
In this situation one can find a quantum state near to the
approximation by renormalizing the positive part of the approximation.
For the applications of tomography in this paper, however, this issue
may safely be disregarded, as non-positive approximations of density
operators will still provide valid approximations to the quantities we
are interested in.

An example of an information-complete measurement on a single qubit is
given by the following matrices:
\begin{alignat*}{2}
P_0 & = \begin{pmatrix} 
  \frac{2 + \sqrt{2}}{8} & \frac{1 + i}{8}\\[1mm]
  \frac{1 - i}{8} & \frac{2 - \sqrt{2}}{8}
\end{pmatrix},
\quad
&
P_1 & = \begin{pmatrix} 
  \frac{2 - \sqrt{2}}{8} & \frac{1 - i}{8}\\[1mm]
  \frac{1 + i}{8} & \frac{2 + \sqrt{2}}{8}
\end{pmatrix},\\[3mm]
P_2 & = \begin{pmatrix} 
  \frac{2 + \sqrt{2}}{8} & \frac{-1 - i}{8}\\[1mm]
  \frac{-1 + i}{8} & \frac{2 - \sqrt{2}}{8}
\end{pmatrix},
\quad
&
P_3 & = \begin{pmatrix} 
  \frac{2 - \sqrt{2}}{8} & \frac{-1 + i}{8}\\[1mm]
  \frac{-1 - i}{8} & \frac{2 + \sqrt{2}}{8}
\end{pmatrix}.
\end{alignat*}
This is not an optimal information-complete measurement, but it has
the advantage of being simple to describe and can be implemented
exactly by a quantum circuit composed of Hadamard, controlled-not, and
$\pi/8$-phase gates, and measurement in the standard basis.
The corresponding set $\{M_0,M_1,M_2,M_3\}$ described above is given
by
\begin{alignat*}{2}
M_0 & = \begin{pmatrix} 
  \frac{1 + \sqrt{2}}{2} & 1 + i\\[1mm]
  1 - i & \frac{1 - \sqrt{2}}{2}
\end{pmatrix},
\quad
&
M_1 & = \begin{pmatrix} 
  \frac{1 - \sqrt{2}}{2} & 1 - i\\[1mm]
  1 + i & \frac{1 + \sqrt{2}}{2}
\end{pmatrix},\\[3mm]
M_2 & = \begin{pmatrix} 
  \frac{1 + \sqrt{2}}{2} & -1-i\\[1mm]
  -1+i & \frac{1 - \sqrt{2}}{2}
\end{pmatrix},
\quad
&
M_3 & = \begin{pmatrix} 
  \frac{1 - \sqrt{2}}{2} & -1+i\\[1mm]
  -1-i & \frac{1 + \sqrt{2}}{2}
\end{pmatrix}.
\end{alignat*}
It holds that $\norm{M_a}_1 = \sqrt{10} < 4$ for 
$a\in\Gamma = \{0,1,2,3\}$.

An information-complete measurement for $k$ qubits may be obtained by
taking tensor products of the above matrices.
More specifically, for each $x\in\Gamma^k$, let us define
$2^k\times 2^k$ matrices $P_x$ and $M_x$ as
\[
P_x = P_{x_1}\otimes\cdots\otimes P_{x_k}
\quad\quad\text{and}\quad\quad
M_x = M_{x_1}\otimes\cdots\otimes M_{x_k}.
\]
Then $\{P_x\,:\,x\in\Gamma^k\}$ is an information-complete
measurement, and its corresponding set is given by $\{M_x\,:\,x\in\Gamma^k\}$.
By the multiplicativity of the trace norm, it holds that
$\norm{M_x}_1 = 10^{k/2} < 4^k$ for every $k$.

Now, let us suppose that $\rho$ is a quantum state on $k$ qubits, and
tomography (using the measurements just described) is performed on $N$
copies of $\rho$.
More precisely, the measurement $\{P_x\}$ is performed independently
on each of the $N$ copies of $\rho$, a probability distribution
$q:\Gamma^k\rightarrow[0,1]$ is taken to be the frequency distribution
of the outcomes, and an approximation
\[
H = \sum_{x\in\Gamma^k} q(x) M_x
\]
to $\rho$ is computed.
We require a bound on the accuracy of this approximation.
Of course, nothing can be said in the worst case, as any sequence of
measurement outcomes could occur with very small probability in
general.

\begin{lemma}\label{lem:tomography} 
For any choice of $\varepsilon > 0$, taking
$N \geq 2^{10k}/\varepsilon^3$ will guarantee that with
probability at least $1 - \varepsilon$, the estimate $H$ satisfies
$\norm{\rho - H}_1 < \varepsilon$.
\end{lemma}
\begin{proof} For any $\delta>0$, and any fixed choice of $x\in\Gamma^k$, it
    follows from Hoeffding's inequality that
    \[
    \op{Pr}\left[\abs{q(x) - p(x)} \geq \delta\right]
    \leq 2\exp\(-2 N \delta^2\).
    \]
    By the union bound it follows that
    \[
    \op{Pr}\left[\norm{q - p}_1 \geq 4^{k}\delta\right]\leq
    \op{Pr}\left[\abs{q(x) - p(x)} \geq \delta\;
      \text{for at least one $x\in\Gamma^k$}\right]
    \leq 2^{2k + 1} \exp\(-2 N \delta^2\).
    \]
    Setting $\delta = \varepsilon/16^{k}$ and using the inequality
    $e^{-\alpha} < 1/\alpha$ for all $\alpha>0$, we have
    \[
    \op{Pr}\left[\norm{q - p}_1 \geq \varepsilon/4^k\right]
    \leq 2^{2k + 1} \exp\(-2^{2k+1} /\varepsilon \) < \varepsilon.
    \]
    It follows that
    \[
    \op{Pr}[\norm{\rho - H}_1 \geq \varepsilon] \leq \op{Pr}[\norm{q-p}_1 \geq \varepsilon/4^k  ] < \varepsilon.
    \]
\end{proof}

The notion of \emph{quantum process tomography} has also been
considered, where a quantum measurement or channel is approximated
through many independent evaluations of an appropriate sort (see for example~\cite{NielsenC00}).
In this paper, however, it is not necessary to consider this sort of
tomography as being any different from state tomography.
Specifically, we will approximate channels (and measurements, modeled
as channels) by evaluating them on maximally entangled states,
followed by ordinary quantum state tomography on the
normalized Choi-Jamio{\l}kowski representations that result.

\subsection{Finite quantum de Finetti theorem} \label{sec:de-Finetti}

Suppose that $\reg{Q}_1, \dots, \reg{Q}_N$ are all $k$-qubit registers. A state on $(\reg{Q}_1, \dots, \reg{Q}_N)$ is called symmetric
if it is invariant under any permutation of its registers. For instance, any product state of the form $\rho^{\otimes N}$ is symmetric. Any convex combination of such states is symmetric as well. Note, however, that there are symmetric states that cannot be written as a convex combination of symmetric product states as above; as an example consider the following state 
\[
\ket{\psi}=\frac{1}{\sqrt{2^k}} \sum_{x\in \Sigma^k} \ket{x} \otimes \cdots \otimes \ket{x}.
\]
Nevertheless, by tracing out any subsystem of $\ket{\psi}$ the resulting reduce density matrix is in the convex hall of symmetric product states. The following theorem generalizes this observation.

\begin{theorem} \label{thm:de-Finetti} (Finite quantum de Finetti theorem
  \cite{KonigR05,ChristandlKMR07}) Suppose that $\rho_{N+m}$ is a symmetric state over registers $(\reg{Q}_1, \dots, \reg{Q}_{N+m})$, and let $\rho_N = \tr_{\Q_{N+1}\cdots \Q_{N+m}} (\rho_{N+m})$. Then there exist states $\xi_j$ and probability vector $p_j$ such that 
  \[
  \norm{  \rho_N -\sum_{j} p_j \xi_j^{\otimes N}   }_1 \leq \frac{N}{N+m}2^{k+1} .
  \]
\end{theorem}

\subsection{Quantum interactive proofs}
\label{sec:QIP}

Quantum interactive proof systems are a natural quantum analogue of
ordinary, classical interactive proof systems, where the prover and
verifier may process and exchange quantum information.
We will only consider quantum interactive proof systems having an even
number of messages in this paper, so for simplicity we will restrict
our discussion to this case.

For $t$ being a function of the form $t:\natural\rightarrow\natural$,
we define a $t$-round (or $(2t)$-message) quantum verifier $V$ to be a
collection of quantum circuits
\[
V = \left\{V_{x,j}\,:\,x\in\Sigma^{\ast},\,\;0\leq j \leq
t(\abs{x})\right\},
\]
which can be generated in polynomial-time given $x$ and $j$.
We will generally write $t$ rather than $t(\abs{x})$ hereafter in this
paper, keeping in mind that $t$ might vary with the input length.
We assume that the verifier's circuits are composed of standard
unitary quantum gates (controlled-not, Hadamard, and $\pi/8$-phase
gates, let us say), as well as ancillary and erasure gates.
Included in the description of these circuits is a specification of
which input and output qubits are to be considered 
\emph{private memory qubits} and which are considered \emph{message}
qubits.
The message qubits refer to qubits that are sent to or received from a
prover (to be described shortly).
The following properties are required of the circuits describing a
verifier:
\begin{mylist}{\parindent}
\item[1.]
  For each $x$, the circuit $V_{x,0}$ takes no input qubits, and the
  circuit $V_{x,t}$ produces a single output qubit (called the
  \emph{acceptance qubit}).
  
\item[2.]
  There exist functions $v_1,v_2,\ldots$ such that
  $V_{x,j-1}$ outputs $v_j(\abs{x})$ private memory qubits and $V_{x,j}$
  inputs $v_j(\abs{x})$ private memory qubits for $1\leq j\leq t$.
  
\item[3.]
  There exist functions $q_1,q_2,\ldots$ and $r_1,r_2,\ldots$ that
  specify the number of message qubits the verifier sends to or receives
  from the prover on each round, for a given input length.
  More precisely, each circuit $V_{x,j-1}$ outputs $q_j(\abs{x})$ message
  qubits and each circuit $V_{x,j}$ inputs $r_j(\abs{x})$ message
  qubits, for $1\leq j \leq t$.
\end{mylist}
Similar to the function $t$, we will often omit the argument $\abs{x}$
from the functions $v_j$, $q_j$, and $r_j$ for the sake of readability.
When it is convenient, we will refer to the message qubits sent from
the verifier to the prover as \emph{question qubits} and qubits sent
from the prover to the verifier as \emph{response qubits}.

A $t$-round (or $(2t)$-message) prover is defined in a similar way to
a $t$-round verifier, but no computational restrictions are made.
Specifically, a $t$-round prover is a collection of quantum channels
\[
P = \left\{P_{x,j}\,:\,x\in\Sigma^{\ast},\:\,1\leq j \leq
t(\abs{x})\right\}.
\]
Again, the input and output qubits of these channels are specified as
private memory qubits or message qubits.
When a particular prover $P$ is considered to interact with a given
verifier $V$, one naturally assumes that they agree on the number of
messages and the number of qubits sent in each message, as suggested
by Figure~\ref{fig:quantum-interactive-proof}. But given the verifier there is no restriction on $p_j$, the number of private memory qubits used by the prover at the $j$-th round. 
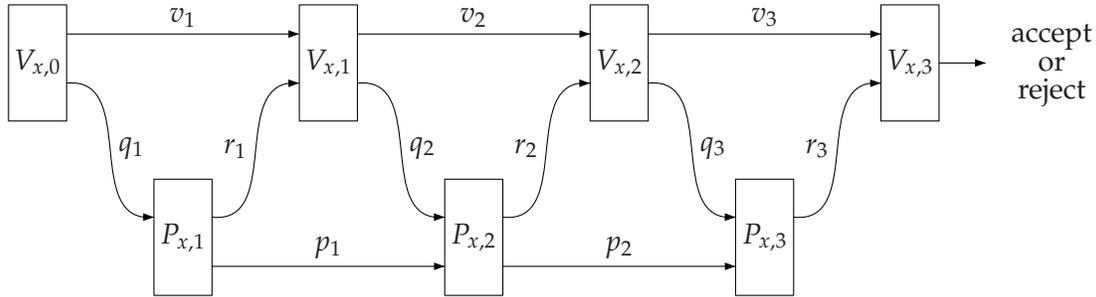
\begin{figure}[t]
  \begin{center}
    \unitlength=2.2pt
    \begin{picture}(200, 50)(25,0)
      \node[Nframe=n,Nw=24](Out)(210,40){%
	\makebox(0,0){%
	  \begin{tabular}{c}
	    accept\\[-1mm] or\\[-1mm] reject
	  \end{tabular}
	}}
      \node[Nw=10,Nh=20,Nmr=0](V1)(35,40){$V_{x,0}$}
      \node[Nw=10,Nh=20,Nmr=0](V2)(85,40){$V_{x,1}$}
      \node[Nw=10,Nh=20,Nmr=0](V3)(135,40){$V_{x,2}$}
      \node[Nw=10,Nh=20,Nmr=0](V4)(185,40){$V_{x,3}$}
      \node[Nw=10,Nh=20,Nmr=0](P1)(60,10){$P_{x,1}$}
      \node[Nw=10,Nh=20,Nmr=0](P2)(110,10){$P_{x,2}$}
      \node[Nw=10,Nh=20,Nmr=0](P3)(160,10){$P_{x,3}$}
      \drawbpedge[syo=-5,eyo=5](V1,24,24,P1,24,-24){$q_1$}
      \drawbpedge[syo=-5,eyo=5](V2,24,24,P2,24,-24){$q_2$} 
      \drawbpedge[syo=-5,eyo=5](V3,24,24,P3,24,-24){$q_3$} 
      \drawbpedge[syo=5,eyo=-5](P1,-24,24,V2,-24,-24){$r_1$} 
      \drawbpedge[syo=5,eyo=-5](P2,-24,24,V3,-24,-24){$r_2$} 
      \drawbpedge[syo=5,eyo=-5](P3,-24,24,V4,-24,-24){$r_3$} 
      \drawedge[syo=5,eyo=5](V1,V2){$v_1$}
      \drawedge[syo=5,eyo=5](V2,V3){$v_2$}
      \drawedge[syo=5,eyo=5](V3,V4){$v_3$}
      \drawedge[syo=-5,eyo=-5](P1,P2){$p_1$}
      \drawedge[syo=-5,eyo=-5](P2,P3){$p_2$}
      \drawedge(V4,Out){}
    \end{picture}
  \end{center}
  \caption{An illustration of an interaction between a prover and
    verifier in a quantum interactive proof system.
    In the picture it is assumed that $t = 3$.
    The labels $v_j$, $p_j$, $q_j$ and $r_j$ on the arrows refer to the
    number of qubits represented by each arrow.}
  \label{fig:quantum-interactive-proof}
\end{figure}
Although $p_j$ could
in principle be unbounded, it is not difficult to show that 
for any choice of verifier and prover there is another prover that simulates the same interaction and uses at
most a polynomial number of private memory qubits (see~\cite{GutoskiW07}).

Now, on a given input string $x$, the prover $P$ and verifier $V$ have
an interaction by composing their circuits/channels as described in
Figure~\ref{fig:quantum-interactive-proof}.
The \emph{maximum acceptance probability} for a given verifier $V$ on
an input $x$ refers to the maximum probability for the circuit
$V_{x,t}$ to output~1, assuming it is measured in the standard basis,
over all choices of a compatible prover $P$.
It is always the case that a maximal probability is achieved by some
prover.

Classes of promise problems may be defined by quantum interactive
proof systems in a variety of ways.
We will delay the definitions of the classes we consider to the
individual sections in which they are discussed.

\section{Two-message quantum interactive proofs with short questions} 
\label{sec:first-message-short}

The first specific variant of quantum interactive proof systems we
consider are those in which just a single round of communication takes
place, with the first message being short (at most logarithmic-length)
and the second message being normal (at most polynomial-length).
In particular, let us say that a $1$-round verifier $V$ is a
$[\log,\poly]$ quantum verifier if the number $q = q_1$ of question
qubits it sends during the first and only round of communication
satisfies $q(n) = O(\log n)$.
For functions of the form $a,b:\natural\rightarrow[0,1]$ we define
$\class{QIP}([\log,\poly],a,b)$ to be the class of all promise
problems $B = (B_{\yes},B_{\no})$ for which there exists a
$[\log,\poly]$ quantum verifier $V$ with completeness and soundness
probability bounds $a$ and $b$, respectively.
In other words, $V$ satisfies the following properties:
\begin{mylist}{\parindent}
\item[1.]
  For every string $x\in B_{\yes}$, there exists a prover $P$
  compatible with $V$ that causes $V$ to accept $x$ with probability
  at least $a(\abs{x})$.
\item[2.]
  For every string $x\in B_{\no}$, and every prover $P$ compatible
  with $V$, it holds that $P$ causes $V$ to accept $x$ with
  probability at most $b(\abs{x})$.
\end{mylist}
For a wide range of choices of $a$ and $b$, these classes coincide
with $\class{QMA}$ as the following theorem states.

\begin{theorem}
  Let $a,b:\natural\rightarrow (0,1)$ be polynomial-time computable
  functions such that $a(n) - b(n) \geq 1/p(n)$ for some polynomial $p$.
  Then $\class{QIP}([\log,\poly],a,b) = \class{QMA}$.
\end{theorem}

\begin{proof}
  It is clear that $\class{QMA}\subseteq\class{QIP}([\log,\poly],a,b)$
  for any choice of $a$ and $b$ that satisfy the conditions of the
  theorem, so our goal is to prove the reverse containment.
  
  Let $B = (B_{\yes},B_{\no})$ be a promise problem in 
  $\class{QIP}([\log,\poly],a,b)$, and let $V$ be a $[\log,\poly]$
  verifier that witnesses this fact.
  We write $q$ (as above) to denote the number of question qubits the
  verifier $V$ sends, and write $r$ to denote the number of response
  qubits $V$ receives.
  As $V$ is a $[\log,\poly]$ verifier it holds that $q(n) = O(\log n)$.
  For a fixed input $x$, we will write $\Q = \complex(\Sigma^q)$ to
  denote the \emph{question space} and $\R = \complex(\Sigma^r)$ to
  denote the \emph{response space} for~$V$, corresponding to the
  question and response qubits in the obvious way.

  Our goal is to prove that $B\in\class{QMA}$, and to do this we will
  define a verification procedure (to be referred to as \emph{Arthur})
  that demonstrates this fact.
  Suppose $P$ is a prover that interacts with $V$.
  For a fixed input string $x$, the action of $P$ may be identified
  with a quantum channel $\Phi\in\channel{\Q,\R}$, and any such
  channel defines a quantum state $\rho\in\density{\R\otimes\Q}$
  according to  its normalized Choi-Jamio{\l}kowski representation
  \eqref{eq:normalized-Choi}.
  We will define Arthur so that he expects to receive many independent
  copies of this state.
  He will check its validity using quantum state tomography, and will
  use the state to apply the mapping $\Phi$ himself through
  post-selection.

  More specifically, we define Arthur so that he performs the
  following actions:
  \begin{mylist}{\parindent}
  \item[1.]
    Input $N+m$ registers
    $(\reg{R}_1,\reg{Q}_1),\ldots,(\reg{R}_{N+m},\reg{Q}_{N+m})$,
    where $N$ and $m$ are polynomials in the input length $n$ to be
    specified below.
    
  \item[2.]
    Randomly permute the pairs 
    $(\reg{R}_1,\reg{Q}_1),\ldots,(\reg{R}_{N+m},\reg{Q}_{N+m})$,
    according to a uniformly chosen permutation $\pi\in S_{N+m}$, and
    discard all but the first $N+1$ pairs.
    
  \item[3.]
    Perform quantum state tomography on the registers 
    $(\reg{Q}_2,\ldots,\reg{Q}_{N+1})$, and \emph{reject} if the
    resulting approximation is not within trace-distance $\delta/2$ of
    the completely mixed state $\I/2^{q}$, for $\delta$ to be
    specified below.
    
  \item[4.]
    Simulate the original protocol $(P,V)$ by post-selection using the
    register pair $(\reg{R}_1,\reg{Q}_1)$.
    \emph{Reject} if the post-selection fails, and otherwise
    \emph{accept} or \emph{reject} as the outcome of the proof system
    dictates.
  \end{mylist}
  
  \noindent
  To specify $N$, $m$ and $\delta$, we first set
  \[
  \varepsilon = \frac{1}{p 4^{q+1}}
  \]
  for $p$ being the polynomial whose reciprocal separates the
  completeness and soundness probability bounds $a$ and $b$.
  Now set
  \[
  \delta = \frac{\varepsilon^2}{4}, \quad\quad
  N = \frac{2^{10 q}}{(\delta/2)^3}  \quad\quad\text{and}\quad\quad
  m = \frac{2 N 4^q}{\varepsilon}.
  \]
  Given that $q$ is logarithmic, it holds that $N$, $m$,
  $1/\varepsilon$ and $1/\delta$ are polynomially bounded.

  Suppose first that $x\in B_{\yes}$, which implies that there
  exists a prover $P$ that causes $V$ to accept $x$ with probability at
  least $a$.
  Let $\Phi$ denote the quantum channel that describes the behavior of
  $P$, and let $\rho$ be the normalized Choi-Jamio{\l}kowski
  representation of $\Phi$ as described in \eqref{eq:normalized-Choi}.
  Then for each of the register pairs $(\reg{R}_j,\reg{Q}_j)$ being
  prepared independently in the state $\rho$, it holds that Arthur
  rejects in step 3 with probability at most $\delta/2$ (Lemma~\ref{lem:tomography}), and accepts
  in step 4 with probability at least $a/4^q$ (conditioned on not
  having rejected in step 3).
  Arthur therefore accepts with probability at least
  \[
  \(1 - \frac{\delta}{2}\)\frac{a}{4^q} 
  > \frac{a}{4^q} - \varepsilon.
  \]
  
  Now let us suppose that $x\in B_\textup{no}$.
  We first consider the situation in which the state of the registers
  $(\reg{Q}_1,\ldots,\reg{Q}_{N+1})$ at the beginning of step 3 has
  the form
  \[
  \xi^{\otimes (N+1)}
  \]
  for some density operator $\xi\in\density{\Q}$.
  There are two cases to consider: one is that
  $\norm{\xi - \I/2^q}_1 < \delta$ and the other is that
  $\norm{\xi - \I/2^q}_1 \geq \delta$.
  If it is the case that $\norm{\xi - \I/2^q}_1 < \delta$, then
  by the Fuchs-van de Graaf inequalities \cite{FuchsvdG99} there
  must exist a state $\rho\in\density{\R\otimes\Q}$ satisfying
  $\tr_{\R}(\rho) = \I/2^q$ that is within trace distance
  $\varepsilon$ of the state of $(\reg{R}_1,\reg{Q}_1)$. To be more precise, consider a fixed purification $\ket{\psi}$ of the state of $(\reg{R}_1, \reg{Q}_1)$ with an auxiliary register $\reg{E}$. Since $\xi=\tr_{\R_1\E_1} (\ket{\psi}\!\bra{\psi})$ has a high fidelity with $\I/2^q$, and due to the characterization of fidelity in terms of purifications, there exists a purification of $\I/2^q$ over $(\reg{R}, \reg{Q}, \reg{E})$ that has a large overlap with $\ket{\psi}$. Then  $\rho$ can be chosen as the reduce density matrix of this pure state over $(\reg{R}, \reg{Q})$.
  Now given that $x\in B_{\textup{no}}$, the state $\rho$ would cause
  acceptance in step 4 with probability at most $b/4^q$, and therefore
  acceptance may occur in the case at hand with probability at most
  $b/4^q + \varepsilon$.
  If, on the other hand, it holds that $\norm{\xi - \I/2^q}_1 \geq
  \delta$, then rejection must occur in step 3 with probability at
  least $1 - \delta/2$, so Arthur accepts with probability
  at most $\delta/2$ (which of course is smaller than 
  $b/4^q + \varepsilon$).
  Thus, in both cases, acceptance occurs with probability at most
  $b/4^q + \varepsilon$.
  It follows that if the registers $(\reg{Q}_1,\ldots,\reg{Q}_{N+1})$
  are, at the beginning of step 3, in any state of the form
  \begin{equation} \label{eq:convex-product-state}
  \sum_{j} p_j \xi_j^{\otimes (N+1)}
  \end{equation}
  (i.e., a convex combination of states of the form just discussed),
  acceptance may occur with probability at most 
  $b/4^q + \varepsilon$.
  Finally, by the finite quantum de Finetti theorem (Theorem~\ref{thm:de-Finetti})
  it holds that the state of
  $(\reg{Q}_1,\ldots,\reg{Q}_{N+1})$ after step 2, is within trace-distance
  $\varepsilon$ of a state of the form \eqref{eq:convex-product-state}, and therefore
  the probability of acceptance is at most $b/4^q + 2\varepsilon$ in
  the general case.
  
  Given that $a/4^q - \varepsilon$ and $b/4^q + 2\varepsilon$ are
  efficiently computable and separated by the reciprocal of a
  polynomial, it holds that $B$ is in $\class{QMA}$ as claimed.
\end{proof}

\section{Quantum interactive proofs with only short messages} 
\label{sec:all-messages-short}

Next we consider quantum interactive proof systems restricted so that
the total number of qubits exchanged by the prover and verifier is
logarithmic.
We prove that any problem having such a quantum interactive proof
system is contained in \class{BQP}.
This fact represents a significant generalization of the equality
$\class{QMA}_{\log} = \class{BQP}$ proved in \cite{MarriottW05}.
Like the result of the previous section, our proof of this fact is
based on quantum state tomography.
In addition we will make use of the \emph{quantum games} framework of
\cite{GutoskiW07}.

It is clear that any quantum interactive proof system allowing at most
a logarithmic number of qubits to be exchanged can be simulated by one
in which a logarithmic number of single qubit messages are
permitted, because any number of these messages could consist of
meaningless ``dummy'' qubits that are interspersed with the qubits
sent by the other party.
To be more precise, let $t(n) = O(\log n)$ and consider a $t$-round
quantum interactive proof system in which each message consists of a
single qubit (i.e., $q_1 = r_1 = \cdots = q_t = r_t = 1$).
We will write $\class{QIP}_{\log}(a,b)$ to denote the class of
problems having quantum interactive proof systems of this sort having
completeness and soundness probability bounds $a$ and $b$,
respectively.
As the following theorem states, this model offers no computational
advantage over $\class{BQP}$.

\begin{theorem}
  Let $a,b:\natural\rightarrow (0,1)$ be polynomial-time computable
  functions such that $a(n) - b(n) \geq 1/p(n)$ for some polynomial
  $p$.
  Then $\class{QIP}_{\log}(a,b) = \class{BQP}$.
\end{theorem}

\begin{proof}
It is clear that $\class{BQP}\subseteq\class{QIP}_{\log}(a,b)$, and so
it remains to prove the reverse containment.
To this end let $B=(B_{\yes},B_{\no})$ be a promise problem in 
$\class{QIP}_{\log}(a,b)$, and let $V$ be a verifier that witnesses
this fact.
As above, let $t(n) = O(\log n)$ denote the number of rounds of
communication this verifier exchanges with any compatible prover.
For a fixed input string $x$, we will write $\Q_1,\ldots,\Q_t$ to
denote copies of the Hilbert spaces $\complex(\Sigma)$ associated with
the $t$ single-qubit messages that $V$ sends to a given prover $P$,
and we will write $\R_1,\ldots,\R_t$ to denote copies of the same
space $\complex(\Sigma)$ corresponding to the response qubits of $P$.

The action of $V$, on a given input string $x$, is determined by $t+1$
quantum circuits $V_{x,0},\ldots,V_{x,t}$ as defined in
Section~\ref{sec:preliminaries}.
Figure~\ref{fig:qubit-message-protocol} illustrates an interaction
between $V$ and a prover $P$ for the case that $t=4$.
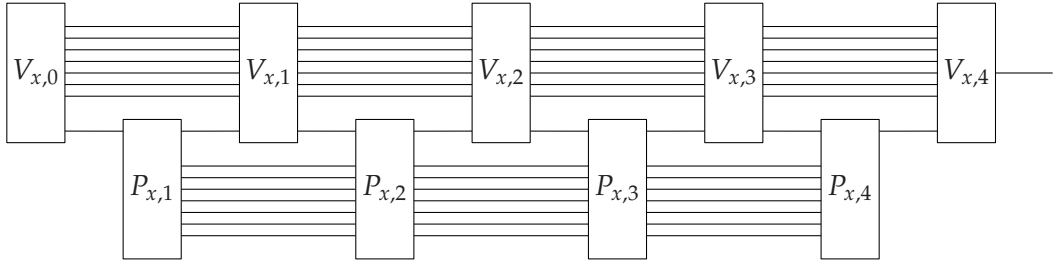
\begin{figure}[t]
  \begin{center}
    \unitlength=0.44pt
    \begin{picture}(1000, 260)(0,20)
      \gasset{Nmr=0,AHnb=0}
      \node[Nw=50,Nh=120](V0)(100,200){$V_{x,0}$}
      \node[Nw=50,Nh=120](V1)(300,200){$V_{x,1}$}
      \node[Nw=50,Nh=120](V2)(500,200){$V_{x,2}$}
      \node[Nw=50,Nh=120](V3)(700,200){$V_{x,3}$}
      \node[Nw=50,Nh=120](V4)(900,200){$V_{x,4}$}
      \node[Nw=50,Nh=120](P1)(200,100){$P_{x,1}$}
      \node[Nw=50,Nh=120](P2)(400,100){$P_{x,2}$}
      \node[Nw=50,Nh=120](P3)(600,100){$P_{x,3}$}
      \node[Nw=50,Nh=120](P4)(800,100){$P_{x,4}$}
      \drawedge[syo=40,eyo=40](V0,V1){}
      \drawedge[syo=30,eyo=30](V0,V1){}
      \drawedge[syo=20,eyo=20](V0,V1){}
      \drawedge[syo=10,eyo=10](V0,V1){}
      \drawedge(V0,V1){}
      \drawedge[syo=-10,eyo=-10](V0,V1){}
      \drawedge[syo=-20,eyo=-20](V0,V1){}
      \drawedge[syo=40,eyo=40](V1,V2){}
      \drawedge[syo=30,eyo=30](V1,V2){}
      \drawedge[syo=20,eyo=20](V1,V2){}
      \drawedge[syo=10,eyo=10](V1,V2){}
      \drawedge(V1,V2){}
      \drawedge[syo=-10,eyo=-10](V1,V2){}
      \drawedge[syo=-20,eyo=-20](V1,V2){}
      \drawedge[syo=40,eyo=40](V2,V3){}
      \drawedge[syo=30,eyo=30](V2,V3){}
      \drawedge[syo=20,eyo=20](V2,V3){}
      \drawedge[syo=10,eyo=10](V2,V3){}
      \drawedge(V2,V3){}
      \drawedge[syo=-10,eyo=-10](V2,V3){}
      \drawedge[syo=-20,eyo=-20](V2,V3){}
      \drawedge[syo=40,eyo=40](V3,V4){}
      \drawedge[syo=30,eyo=30](V3,V4){}
      \drawedge[syo=20,eyo=20](V3,V4){}
      \drawedge[syo=10,eyo=10](V3,V4){}
      \drawedge(V3,V4){}
      \drawedge[syo=-10,eyo=-10](V3,V4){}
      \drawedge[syo=-20,eyo=-20](V3,V4){}
      \drawedge[syo=-50,eyo=50](V0,P1){}
      \drawedge[syo=50,eyo=-50](P1,V1){}
      \drawedge[syo=-50,eyo=50](V1,P2){}
      \drawedge[syo=50,eyo=-50](P2,V2){}
      \drawedge[syo=-50,eyo=50](V2,P3){}
      \drawedge[syo=50,eyo=-50](P3,V3){}
      \drawedge[syo=-50,eyo=50](V3,P4){}
      \drawedge[syo=50,eyo=-50](P4,V4){}
      \drawedge[syo=-40,eyo=-40](P1,P2){}
      \drawedge[syo=-30,eyo=-30](P1,P2){}
      \drawedge[syo=-20,eyo=-20](P1,P2){}
      \drawedge[syo=-10,eyo=-10](P1,P2){}
      \drawedge(P1,P2){}
      \drawedge[syo=10,eyo=10](P1,P2){}
      \drawedge[syo=20,eyo=20](P1,P2){}
      \drawedge[syo=-40,eyo=-40](P2,P3){}
      \drawedge[syo=-30,eyo=-30](P2,P3){}
      \drawedge[syo=-20,eyo=-20](P2,P3){}
      \drawedge[syo=-10,eyo=-10](P2,P3){}
      \drawedge(P2,P3){}
      \drawedge[syo=10,eyo=10](P2,P3){}
      \drawedge[syo=20,eyo=20](P2,P3){}
      \drawedge[syo=-40,eyo=-40](P3,P4){}
      \drawedge[syo=-30,eyo=-30](P3,P4){}
      \drawedge[syo=-20,eyo=-20](P3,P4){}
      \drawedge[syo=-10,eyo=-10](P3,P4){}
      \drawedge(P3,P4){}
      \drawedge[syo=10,eyo=10](P3,P4){}
      \drawedge[syo=20,eyo=20](P3,P4){}
      \node[Nframe=n](Out)(1000,200){}
      \drawedge(V4,Out){}
    \end{picture}
    \caption{Illustration of a quantum interactive proof in which the
      messages are single bits.}
    \label{fig:qubit-message-protocol}
  \end{center}
\end{figure}
Now consider the channel $\Phi$ obtained from the circuits
$V_{x,0},\ldots,V_{x,t}$ by setting all of the response qubits the verifier
receives from the prover as input qubits and setting all of the
question qubits sent by the verifier to the prover as output qubits.
More precisely, $\Phi$ maps states on the space
$\R_1\otimes\cdots\otimes\R_t$ to states on the space
$\A\otimes\Q_1\otimes\cdots\otimes\Q_t$, where $\A$ denotes the
single-qubit space associated with the acceptance qubit.
Figure~\ref{fig:qubit-message-protocol-channel} illustrates this
channel for the protocol pictured in
Figure~\ref{fig:qubit-message-protocol}.
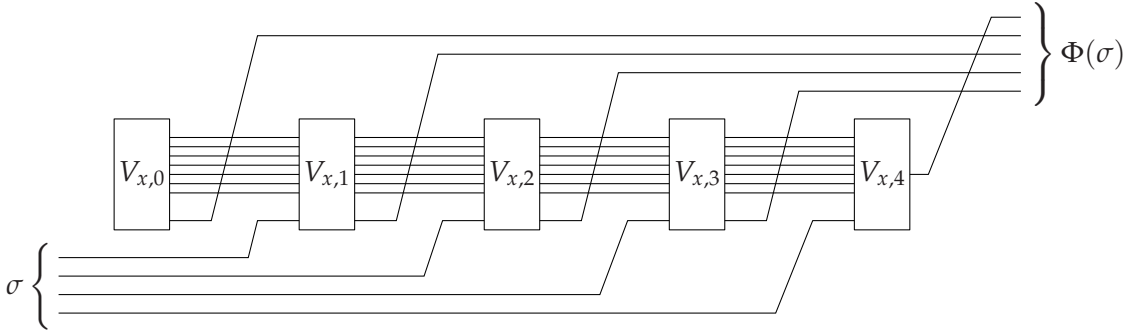
\begin{figure}[t]
  \begin{center}
    \unitlength=0.35pt
    \begin{picture}(1160, 350)(-26,27)
      \gasset{Nmr=0,AHnb=0}
      \node[Nw=60,Nh=120](V0)(100,200){$V_{x,0}$}
      \node[Nw=60,Nh=120](V1)(300,200){$V_{x,1}$}
      \node[Nw=60,Nh=120](V2)(500,200){$V_{x,2}$}
      \node[Nw=60,Nh=120](V3)(700,200){$V_{x,3}$}
      \node[Nw=60,Nh=120](V4)(900,200){$V_{x,4}$}
      \drawedge[syo=40,eyo=40](V0,V1){}
      \drawedge[syo=30,eyo=30](V0,V1){}
      \drawedge[syo=20,eyo=20](V0,V1){}
      \drawedge[syo=10,eyo=10](V0,V1){}
      \drawedge(V0,V1){}
      \drawedge[syo=-10,eyo=-10](V0,V1){}
      \drawedge[syo=-20,eyo=-20](V0,V1){}
      \drawedge[syo=40,eyo=40](V1,V2){}
      \drawedge[syo=30,eyo=30](V1,V2){}
      \drawedge[syo=20,eyo=20](V1,V2){}
      \drawedge[syo=10,eyo=10](V1,V2){}
      \drawedge(V1,V2){}
      \drawedge[syo=-10,eyo=-10](V1,V2){}
      \drawedge[syo=-20,eyo=-20](V1,V2){}
      \drawedge[syo=40,eyo=40](V2,V3){}
      \drawedge[syo=30,eyo=30](V2,V3){}
      \drawedge[syo=20,eyo=20](V2,V3){}
      \drawedge[syo=10,eyo=10](V2,V3){}
      \drawedge(V2,V3){}
      \drawedge[syo=-10,eyo=-10](V2,V3){}
      \drawedge[syo=-20,eyo=-20](V2,V3){}
      \drawedge[syo=40,eyo=40](V3,V4){}
      \drawedge[syo=30,eyo=30](V3,V4){}
      \drawedge[syo=20,eyo=20](V3,V4){}
      \drawedge[syo=10,eyo=10](V3,V4){}
      \drawedge(V3,V4){}
      \drawedge[syo=-10,eyo=-10](V3,V4){}
      \drawedge[syo=-20,eyo=-20](V3,V4){}
      \drawline(130,150)(175,150)(225,350)(1050,350)
      \drawline(330,150)(375,150)(420,330)(1050,330)
      \drawline(530,150)(575,150)(615,310)(1050,310)
      \drawline(730,150)(775,150)(810,290)(1050,290)
      \drawline(270,150)(225,150)(215,110)(10,110)
      \drawline(470,150)(425,150)(405,90)(10,90)
      \drawline(670,150)(625,150)(595,70)(10,70)
      \drawline(870,150)(825,150)(785,50)(10,50)
      \drawline(930,200)(950,200)(1018,370)(1050,370)
      \put(-20,80){\makebox(0,0){$\sigma\left\{\rule{0mm}{7mm}\right.$}}
      \put(1110,330){\makebox(0,0){%
          $\left.\rule{0mm}{7.5mm}\right\}\Phi(\sigma)$}}
    \end{picture}
    \caption{The channel $\Phi$ associated with the quantum interactive
      proof from Figure~\ref{fig:qubit-message-protocol}.}
    \label{fig:qubit-message-protocol-channel}
  \end{center}
\end{figure}

Next, let
\[
\rho = \frac{1}{2^t}\sum_{y,z\in\Sigma^t}
\Phi(\ket{y}\!\bra{z})\otimes \ket{y}\!\bra{z}
\]
be the normalized Choi-Jamio{\l}kowski representation of $\Phi$.
The state $\rho$ is obviously efficiently preparable given a
description of $V$.
By independently preparing
$N = 2^{10(2t+1)}/\varepsilon^3$ copies of $\rho$, for $\varepsilon > 0$ to
be specified later, and performing quantum state tomography, one
obtains a Hermitian operator $H$ on
$\A\otimes\R_1\otimes\cdots\otimes\R_t\otimes\Q_1\otimes\cdots\otimes\Q_t$
that satisfies $\norm{H - \rho}_1 < \varepsilon$ with probability at
least $1 - \varepsilon$.
Let us also define
\[
\rho_1 = \(\bra{1}\otimes\I\) \rho \(\ket{1}\otimes\I\)
\quad\quad\text{and}\quad\quad
H_1 = \(\bra{1}\otimes\I\) H \(\ket{1}\otimes\I\)
\]
to denote the projection of these operators on the subspace in which the
qubit $\A$ is $\ket{1}$ (corresponding to accept).

Using the terminology of~\cite{GutoskiW07}, $\rho_1$ is a \emph{co-strategy} which describes the verifier's action, and the prover optimizes the acceptance probability corresponding to $\rho_1$ over all \emph{strategies}. (Strategies are defined similar to co-strategies as above but with respcet to the prover's action.) Given any strategy $X$ of the prover, the acceptance probability is proportional to the inner product of $\rho_1$ and $X$. More precisely, the maximum acceptance probability is equal to 
\begin{align*}
  \text{maximize:}\quad & 2^t \ip{\rho_1}{X} \\
  \text{subject to:}\quad & 
  X\in\S_t
\end{align*}
where $\S_t\subset\pos{\R_1\otimes\cdots\otimes\R_t\otimes\Q_1\otimes
  \cdots\otimes\Q_t}$ is the space of all strategies. It is shown in~\cite{GutoskiW07} that $\S_t$ is characterizes as $\S_0 = 1$ and
\[
\S_j = \left\{
  X\geq 0\,:\,\tr_{\R_j}(X) =
  Y\otimes\I_{\Q_j},\,Y\in\S_{j-1}\right\}
\]
for $j \geq 1$. This characterization of $\S_t$ turns the above optimization problem to a semidefinite program. So we just need to replace $\rho_1$ with its approximation $H_1$.

It is clear that $\tr(X) = 2^t$ for every $X\in\S_t$, and therefore
\[
\abs{2^t \ip{\rho_1}{X} - 2^t\ip{H_1}{X}}
\leq 2^t \norm{X} \norm{\rho_1 - H_1}_1 \leq 4^t \norm{\rho_1 - H_1}_1 \leq 4^t \norm{\rho - H}_1
\]
for every $X\in\S_t$.
By taking 
\[
\varepsilon = \frac{1}{4^{t+1}p}
\]
for instance, one may therefore distinguish the cases $x\in B_{\yes}$
and $x\in B_{\no}$ with probability $1 - \varepsilon$, by solving the semidefinite program described above. (Semidefinite programs can be efficiently solved up to an inverse polynomial accuracy.)
\end{proof}

We note that precisely the same argument allows one to conclude that
\emph{quantum refereed games}, as defined in \cite{GutoskiW07},
allowing for at most a logarithmic number of qubits of communication
offer no computational power beyond $\class{BQP}$.
In other words, $\class{QRG}_{\log} = \class{BQP}$, for 
$\class{QRG}_{\log}$ defined appropriately.
The details are left to the reader.

\section{Two-message quantum interactive proofs with short answers} 
\label{sec:second-message-short}

In light of the results of Section~\ref{sec:first-message-short}, one
may ask if two-message quantum interactive proofs with short 
\emph{answers} (as opposed to short \emph{questions}) have the power of
$\class{QMA}$ or even $\class{BQP}$.
If this is true it is likely to be difficult to show: the graph
non-isomorphism problem, which is not known to be in \class{QMA}, has
a simple and well-known classical protocol \cite{GoldreichMW91}
requiring polynomial-length questions and constant-length answers.
(Indeed, every problem in \class{QSZK} has a two-message quantum
interactive proof system with a constant-length message from the
prover to the verifier, for any choice of constant completeness and
soundness errors \cite{Watrous02}.)

We can show, however, that \emph{public-coin} quantum interactive
proofs in which the verifier sends polynomially many random bits to
the prover, followed by a logarithmic-length quantum message response
from the prover, have only the power of \class{BQP}.

Following a similar terminology to the classical case, we refer to a
quantum interactive proof system in which the verifier's messages to
the prover consist of uniformly-generated random bits as
\emph{quantum Arthur--Merlin games}.
Let us write $\class{QAM}([\poly,\log],a,b)$ to denote the class of
promise problems having two-message quantum Arthur--Merlin games
with completeness and soundness probability bounds $a$ and $b$, in
which Merlin's response to Arthur has logarithmic length.

\begin{theorem}
  Let $a,b:\natural\rightarrow (0,1)$ be polynomial-time computable
  functions such that $a(n) - b(n) \geq 1/p(n)$ for some
  polynomial-bounded function $p$.
  Then $\class{QAM}([\poly,\log],a,b) = \class{BQP}$.
\end{theorem}

\begin{proof}
  Assume that $B$ is a promise problem in $\class{QAM}([\poly,\log],a,b)$,
  and consider a choice of Arthur that witnesses this fact.
  For $r(n) = O(\log n)$, and for any choice of an input string $x$,
  Arthur chooses a random string $y$ with length polynomial in
  $\abs{x}$, and then measures $r = r(\abs{x})$ qubits sent by Merlin
  with respect to some binary-valued measurement $\{P^{x,y}_0,P^{x,y}_1\}$
  that depends on~$x$ and $y$. Thus, assuming that the randomly chosen string is $y$, the maximum acceptance probability of Arthur is equal to the spectral norm of $P^{x,y}_1$. So to find $P^{x,y}_1$ (and its norm) we perform quantum state tomography on the normalized Choi-Jamio{\l}kowski representation of the channel 
  \[
  \Phi_{x,y}(\sigma) = 
  \ip{P^{x,y}_0}{\sigma} \ket{0}\!\bra{0} +
  \ip{P^{x,y}_1}{\sigma} \ket{1}\!\bra{1},
  \]
  which describes Arthur's measurement.
  
  The following algorithm shows that $B\in\class{BQP}$.
  \begin{mylist}{\parindent}
  \item[1.]
    Choose $y$ uniformly at random (just as Arthur does).
  \item[2.]
    Let
    \[
    \varepsilon = \frac{1}{2^{r+3}p}
    \quad\quad\text{and}\quad\quad
    N = \frac{2^{10(r+1)}}{\varepsilon^3}.
    \]
    Prepare $N$ copies of the state $\rho$, defined to be the
    normalized Choi-Jamio{\l}kowski representation of $\Phi_{x,y}$,
    and perform quantum state tomography of $\rho$.
    Let $H$ denote the result. Then by Lemma~\ref{lem:tomography} with probability at least $1-\varepsilon$, $\norm{\rho - H}_1 \leq \varepsilon$.
  \item[3.]
    Compute the value
    \[
    \alpha_y = 2^r \norm{(\bra{1}\otimes\I)H(\ket{1}\otimes\I)}.
    \]
    It can easily be shown that $\alpha_y$ is an approximation of $\norm{P_1^{x,y}}$. If $\alpha_y \geq 1$ then \emph{accept}.
    Otherwise, \emph{accept} with probability $\alpha_y$ and \emph{reject}
    otherwise.
  \end{mylist}

  \noindent
  Since the maximum acceptance probability of Arthur is equal to the expectation value of $\norm{P^{x, y}_1}$ over the random choice of $y$, and with probability $1-\varepsilon$, $\alpha_y$ is within distance $2^r \varepsilon$ of $\norm{P^{x, y}_1}$, the above procedure has acceptance probability within $1/(4p)$ of the
  maximum acceptance probability of Arthur. Therefore $B\in\class{BQP}$.
\end{proof}

\section{Open problems} 
\label{sec:open-problems}

Besides the open problems mentioned above, the following two questions have been raised by unknown referees which we leave for future works.
The first one is to find the expressive power of the following model: the verifier and prover send messages to each other with the total number of $O(\log n)$ qubits and at the end the prover sends a polynomial-size message to the verifier. This model contains $\class{QMA}$ and it seems that combining the ideas in Sections~\ref{sec:first-message-short} and \ref{sec:all-messages-short} the other direction can also be proved. The other model is an interactive protocol in which the verifier always sends public-coin messages to the prover of the total polynomial-length and the prover replies with qubits of the total logarithmic-length.

\bibliographystyle{alpha}
\bibliography{short-messages}

\end{document}